\documentclass[runningheads]{llncs}
\usepackage{graphicx,bm,amsmath,amssymb,xspace,xcolor}

\spnewtheorem{defn}{Definition}{\bfseries}{}

\newcommand{\R}{\mathbb R}
\newcommand{\vecset}{\mathcal X}
\newcommand{\nset}[1]{\{1,\,\dotsc,\,#1\}}
\newcommand{\ip}[1]{\left\langle#1\right\rangle}
\newcommand{\rpd}{\mathsf{RPDim}}
\newcommand{\revprefdim}{\textsc{RevPrefDim}\xspace}
\newcommand{\msr}{\mathsf{SgnRnk}}
\newcommand{\matsgnrank}{\textsc{MatSgnRnk}\xspace}

\begin{document}
\title{Revealed Preference Dimension\\ via Matrix Sign Rank}
\titlerunning{Revealed Preference Dimension via Matrix Sign Rank}
\author{Shant Boodaghians}
\authorrunning{Shant Boodaghians}
\institute{Department of Computer Science, \\University of Illinois at Urbana-Champaign,
\\\email{boodagh2@illinois.edu}
}

\maketitle
\begin{abstract}
Given a data-set of consumer behaviour, the Revealed Preference Graph succinctly encodes inferred relative preferences between observed outcomes as a directed graph. 
Not all graphs can be constructed as revealed preference graphs when the market dimension is fixed.
This paper solves the open problem of determining exactly which graphs are attainable as revealed preference graphs in $d$-dimensional markets.
This is achieved via an exact characterization which closely ties the feasibility of the graph to the Matrix Sign Rank of its signed adjacency matrix.
The paper also shows that when the preference relations form a partially ordered set with order-dimension $k$, the graph is attainable as a revealed preference graph in a $k$-dimensional market.
\keywords{Revealed Preference  \and Matrix Sign Rank \and Partial Order.}
\end{abstract}
\section{Introduction}
	In standard economic analysis and mechanism design,
	it is often assumed that the agents' valuation functions are known {\em a priori},
	or more commonly, a probability distribution over possible agent types is assumed to be known.
	However, in practice, we may only observe the prices which are set,
	and the subsequent behaviour of the agents.
	Assuming the agents act rationally, 
	and that their utility functions are restricted to some well-defined class,
	information about the relative values attributed to various outcomes may be inferred
	by simply observing the agents' behaviour at various prices.
	This idea was first pioneered by Paul Samuelson in 1938~\cite{Sam38},
	and a large body of work has followed. 
	See~\cite{Var06} for a thorough survey on the subject.
	The concept came to be known as {\em revealed preference}:
	when an agent chooses different outcomes given different prices,
	she is (under some assumptions) revealing that one is preferable to the other.
	

	Though this may seem natural at first, 
	the implementation of these ideas has required much mathematical development for the description, characterization, and computation of the revealed preferences.
	In the model originally studied by Samuelson, 
	the agent is assumed to have an underlying valuation function 
	and a fixed budget.
	She seeks to choose the collection of goods which has the largest value 
	while satisfying her budget constraint at the current prices.
	This is the most common formulation, though others exist (see {\em e.g.} \cite{CES17}).
	This paper, however, deals only with the standard model.
	
	The market consists of $d$ distinct, separable goods, 
	and a collection of goods (or {\em bundle}) is denoted as a vector $\bm x\in\R^{d}_{\geq 0}$, where the $i$-th coordinate of $\bm x$ represents the quantity of the $i$-th good. 
	The prices are linear, and are described by a vector $\bm p\in\R^d_{\geq 0}$ such that the price of a bundle $\bm x$ is given by the inner product $\ip{\bm p,\bm x}$.
	In this model, a bundle $\bm x$ is {\em revealed preferred} 
	to a bundle $\bm y$ if at current prices $\bm p$,
	the agent chose $\bm x$, but $\bm y$ was more affordable, {\em i.e.} 
	$\ip{\bm p,\bm x}\geq \ip{\bm p,\bm y}$.
	Since the agent is assumed to be maximizing value within a budget constraint, $\bm y$ must be less valuable than $\bm x$.
	
	Samuelson originally asked whether revealed preferences may be used to verify whether an agent's behaviour is consistent with the model assumptions:
	if an agent's behaviour is contradictory, then one must conclude that the model assumption is incorrect.
	He first conjectured a simple test, which was confirmed to be correct in special cases by Rose~\cite{Rose58}, but disproved in the general case. 
	Houthakker~\cite{Hou50} proposed a stronger test of ``cyclical consistency'' and proved its correctness non-constructively in our setting.
	The most famous result is given by Afriat~\cite{Afr67}, where he shows a slightly more general result, and gives explicit constructions of the valuation function as a certificate for consistency.
	The notion of cyclical consistency has since been called the ``Strong/Generalized Axiom of Revealed Preference'' (SARP/GARP),
	and is used to this day in many empirical settings~\cite{Gro95,Var07,Var12}, including as bidding rules in combinatorial auction mechanisms, to deter non-truthful bidding practices~\cite{ACM06,HBPZ10,Cra13,AB14}.
	We will not delve into the details of the axioms of revealed preference, but a thorough survey may be found at~\cite{Var06}.
	It is however helpful to be familiar with the concept of a {\em revealed preference graph}, 
	which is defined below.
	A more familiar treatment to economic audiences is given in~\cite{CE16chapter}.

	
	\medskip\noindent\textbf{Revealed Preference Graphs:}
	Recall that in Samuelson's model, 
	an agent may choose goods from the {\em consumption space} $\R^d_{\geq 0}$,
	and seeks to choose the bundle $\bm x^*$ which maximizes her valuation $v(\bm x)$ subject to the budget constraint $\ip{\bm p,\bm x}\leq 1$, up to re-scaling.	
	Suppose now that we make $n$ observations of this agent at {\em different} price points $\bm p_1,\,\bm p_2,\,\dotsc,\,\bm p_n\in\R^d_{\geq 0}$, 
	and that at prices $\bm p_i$, her optimal bundle was $\bm x_i$.
	If her behaviour is rational, then whenever $\ip{\bm p_i,\bm x_i}\geq \ip{\bm p_i,\bm x_j}$, she must value $\bm x_i$ greater than $\bm x_j$ since the latter was affordable when the former was chosen.
	Thus, she has revealed that $\bm x_i$ is preferable to $\bm x_j$.
	
	These preference relations may be modelled as a directed graph:
	let $G$ be a graph on vertex set $\nset n$, and add an edge directed from $i$ to $j$ if $\bm x_i$ is revealed preferred to $\bm x_j$.
	Thus, $(i,j)\in G$ if and only if $\ip{\bm p_i,\bm x_i}\geq \ip{\bm p_i,\bm x_j}$.
	This graph is implicit in the proofs of Afriat, and the notion of cyclical consistency is equivalent to requiring that $G$ not contain any directed cycles. 
	In general, a preference graph inferred from observations need not be acyclic.\bigskip
	

Most uses of revealed preference as bidding rules in combinatorial auctions (cited above) rely on testing properties of this revealed preference graph. 
In past work~\cite{BV15-garp,BV15-disk}, we have asked whether such tests, {\em e.g.} the {\em minimum feedback vertex set}, are efficiently computable.
We concluded that for this one test, its computational complexity is in fact dependent on the market dimension $d$:
when the dimension of the $\bm p_i$ and $\bm x_j$ vectors is fixed, but the number of observations is unbounded, then the class of observable graphs is restricted, and the computational complexity of some problems may depend on the parameter $d$.
For example, it was shown by Deb and Pai~\cite{DP14} that when $d= n$, every directed graph on $n$ vertices is observable over $\R^n$, but that for all fixed $d$, there exist exponentially large graphs which can not be observed in $d$ dimensions.

This past work has led us to asking whether one could characterize the class of preference graphs observable in the market $\R^d$ for some fixed $d$.
In fact, this question had been posed as an open problem by Federico Echenique~\cite{CE16chapter}.

%
%
%
%
	\bigskip
	\noindent\textbf{Question:} For a fixed dimension $d$, which directed graphs may be observed as revealed-preference graphs on bundles in $\R^d$?\bigskip
	
	This paper answers this question by giving an exact characterization of
	all graphs observable over $\R^d$, given in terms of the {\em matrix sign rank} of signed adjacency matrices.
	The notion of matrix sign rank deals with the existence of low-rank matrices whose entries satisfy certain sign constraints;
	see~\cite{BK15} for an introduction.
	More formally, the matrix sign rank of a {\em sign matrix} $S\in \{0,+,-\}^{n\times m}$ is defined as the least rank real-valued matrix $M\in \R^{n\times m}$ whose entries have the same signs as the entries in $S$.
	Matrix sign rank has been influential in many fields, including complexity theory and learning theory. 
	Seminal lower bounds in communication complexity~\cite{For02} and circuit complexity~\cite{RS10} rely on using the matrix sign rank of a problem to measure its hardness, 
	whereas linear classification algorithms benefit from low-dimensional embeddings of classification problems, as given by sign-rank~\cite{AMY16,perceptron,svm}.
	
	\vspace*{-1ex}
	
	\subsubsection*{Summary of the results.} 
	To answer the above question, we introduce the notion of {\em RP dimension},
	defined as follows:
	Given a directed graph $G$, what is the least $d$ such that $G$ may be observed as a preference graph over $\R^d$? 
	
	In this paper, we give an exact characterization of the set of graphs
	with RP dimension $d$, for all $d\geq 1$.
	We show that, for a given graph, it can be realized as preference observations over $\R^d$ if and only if its signed adjacency matrix has sign-rank at most $d+1$. 
	Thus, the RP dimension of the graph $G$ is exactly the sign-rank of its signed adjacency matrix, minus one.
	In fact, this paper shows that determining the RP dimension of directed graphs is equivalent to determining the sign rank of a large class of sign matrices, a problem which is known to be NP-hard~\cite{BK15}.
	
	

	This paper also considers the special case of directed graphs which represent {\em partially ordered sets}, or posets.
	We show that the RP dimension is at most the {\em order-dimension} of the poset, and that this bound is tight when the order-dimension is at most 3.
	However, there exist posets of arbitrarily large order-dimension, which can be realized in $\R^3$.

%
%

\vspace*{-8pt}
	
\subsubsection*{Acknowledgements.} 
I would like to thank Ruta Mehta, Adrian Vetta, and Siddharth Barman, for their insightful discussion in the initial stages of work.

\section{Model, Preliminaries, and Summary of Results}\label{sec:model-and-prelims}
This section formally lays out the concepts introduced above.
As discussed in the introduction, it is assumed that an agent is observed repeatedly in a market.
Faced with price-vector $\bm p\in \R^{d}_{\geq 0}$, she chooses the item which maximizes her valuation subject to a budget constraint. 
Thus, she chooses $\bm x^*$ as the bundle which maximizes $v(\bm x)$ subject to $\ip{\bm p,\bm x}\leq 1$.
Assume that on the $i$-th observation, the agent was faced with prices $\bm p_i$, and chose the bundle $\bm x_i$.
Then we have that $\bm x_i$ is {\em revealed preferred} to $\bm x_j$ if $\ip{\bm p_i,\bm x_i}\geq \ip{\bm p_i,\bm x_j}$,
since $\bm x_j$ must have been affordable when $\bm x_i$ was chosen.
Given a collection of observations $(\bm p_1,\bm x_1),\,(\bm p_2,\bm x_2),\,\dotsc,\,(\bm p_n,\bm x_n)$, 
we may construct a directed {\em preference graph} $G$ on vertex set $\nset n$ with an edge from $i$ to $j$ if $\bm x_i$ is revealed preferred to~$\bm x_j$.

This paper, however, does not deal with constructing preference graphs from data sets, but rather of constructing data sets from preference graphs.
Thus, we introduce the notion of the {\em realization} of a preference graph:

\begin{defn}[RP realization] \label{def:RP-realization}
Let $G=(V,E)$ be a directed graph with vertices labelled $\nset n$, 
and let $\vecset:=\{(\bm p_1,\bm x_1),\,\dotsc,\,(\bm p_n,\bm x_n)\}$ be pairs of vertices in $\R^d$
such that $\bm p_i\geq \bm 0$ for all $i\leq n$. 
Then $\mathcal X$ is said to {\em RP-realize} $G$ if for all $1\leq i,j\leq n$, 
the directed edge $(i,j)$ is present in $G$ if and only if $\ip{\bm p_i,\bm x_i}> \ip{\bm p_i,\bm x_j}$.
\end{defn}

Note that we require strict inequality to induce preference. 
This is purely for mathematical convenience, and is not standard in the definitions of revealed preference.
Since we are only considering the existence of the realization, rather than the realization itself, this is assumed without loss of generality.
We also define a notion of {\em weak} RP realization, which we will show is equivalent.

\begin{defn}[Weak RP realization] \label{def:weak-RP-realization} 
As above, let $G=(V,E)$ be a directed graph with vertices labelled $\nset n$, 
and let $\vecset:=\{(\bm p_1,\bm x_1),\,\dotsc,\,(\bm p_n,\bm x_n)\}$ be pairs of vertices in $\R^d$
such that $\ip{\bm p_i,\bm 1}>0$ for all $i\leq n$, where $\bm 1=(1,1,1,\dotsc)$. 
Then $\mathcal X$ is said to {\em weakly RP-realize} $G$ if for all $1\leq i,j\leq n$, 
the directed edge $(i,j)$ is present in $G$ if and only if $\ip{\bm p_i,\bm x_i}> \ip{\bm p_i,\bm x_j}$.
\end{defn}

The difference between RP realization and weak RP realization is the restriction on the possible $\bm p$ vectors.
We will show below that these two notions of RP realization are equivalent in the following sense:
given a graph $G$ and an integer $d$, there exists an RP realization of $G$ in $\R^d$ if and only if there exists a weak RP realization in $\R^d$.

It is natural to ask whether an RP realization is possible, and whether this depends on the value of $d$.
In fact, it was shown by Dep and Pai~\cite{DP14} that when $d=n$, a realization is always possible. 
(Simply set $\bm x_i$ to be the $i$-th standard basis vector, and set $(\bm p_i)_j=0$ if $(i,j)\in G$, $1$ if $i=j$, and $2$ if $(i,j)\notin G$.) 
We wish to determine the minimum value of $d$ for which this is possible. 
Thus, we introduce the notion of RP dimension.

\begin{defn}[RP dimension] \label{def:RP-dimension}
Let $G=(V,E)$ be a directed graph on $n$ vertices.
Then the {\em RP dimension} of $G$, (denoted $\rpd(G)$) is the minimum $d$ such that 
there exists an RP realization of $G$ in $\R^d$. 
We denote as \revprefdim the computational problem of computing $\rpd(G)$ for a given di-graph $G$.
\end{defn}

As mentioned above, we will show taht the RP dimension does not change if the RP realization is allowed to be weak. 
We will be characterizing the RP dimension of candidate preference graphs by the sign-rank of an associated sign matrix.
We define below the notion of realization for sign matrices, and define sign rank as the minimum rank of a realization:

\begin{defn}[Matrix sign realization]\label{def:MSR-realization} 
Let $M \in \{+,-,0\}^{n\times m}$ be an $n\times m$ matrix
whose entries are given by the symbols $+$, $-$, and $0$.
Let $A\in \R^{n\times m}$ be an $n\times m$ real-valued matrix. 
Then $A$ is a {\em matrix sign realization} of $M$ if $A_{ij}>0$ whenever $M_{ij}=+$, 
$A_{ij}<0$ whenever $M_{ij}=-$, and $A_{ij}=0$ whenever $M_{ij}=0$.
\end{defn}
	
\begin{defn}[Matrix sign-rank]\label{def:matrix-sign-rank} 
Let $M \in \{+,-,0\}^{n\times m}$, then the {\em sign-rank} of $M$,
(denoted $\msr(M)$)
is the minimum $r$ such that there exists a matrix sign realization $A$ of $M$ with rank $r$.
We denote as \matsgnrank the computational problem of finding $\msr(M)$ for a given matrix $M$.
\end{defn}

Finally we will show that when the preference graph is induced by a {\em partially ordered set} or {\em poset},
the RP dimension of the graph is related to properties of the poset.
Partially ordered sets have been the subject of much study, 
and many textbooks on the matter make for a good introduction (see {\em e.g.} \cite{poset}).
Below is a formal definition, included for completeness.

\begin{defn}[Partially-ordered set (Poset)] \label{def:poset}
Let $S$ be some (finite) ground set, 
and let $\succ$ be a transitive, acyclic, and irreflexive, binary relation on $S$. 
That is, for all $a,b,c\in S$, $a \nsucc a$, and if $a\succ b$ and $b\succ c$, then $a\succ c$.%
\footnote{These two conditions imply that the relation must be acyclic.}
Then the pair $(S,\succ)$ is termed a {\em partially-ordered set, or poset}.
The poset may also be seen as a directed graph $G=(S,\succ)$, which must be acyclic and transitively closed.
A~{\em total order} is a poset whose underlying undirected graph is complete. 
Alternatively, a total order is the poset induced by a ranking of the elements of $S$.
\end{defn}

Every partial order is the {\em intersection} of some total orders.
The {\em order dimension} of a poset captures the least number of total orders needed to realize~it. 

\begin{defn}[Order dimension] \label{def:order-dimension}
	Let $(S,\succ)$ be a poset, and \mbox{$\mathcal O=\{\succ_1,\,\dotsc,\,\succ_k\}$} 
	be $k$ distinct total orders on $S$.
	Then $\mathcal O$ {\em realizes} the poset $(S,\succ)$ if for all $a,b\in S$,
	$a\succ b$ if and only if $a\succ_ib$ for all $1\leq i\leq k$.
	The {\em order dimension} of $(S,\succ)$ is the minimum $k$ such that there
	exists a collection of $k$ total orders which realize $\succ$.
\end{defn}

\subsection{Results}\label{sec:results}
We formally state here the results of this paper and outlines of their proofs. 
The proofs in full technical detail will be presented in the next sections.
The main goal is to show that for each directed graph $G$, 
there exists an associated sign matrix $M$ such that $\rpd(G)=\msr(M)-1$.

For simplicity of notation, though, we extend all directed graphs with a fully-dominated and fully-dominating node as follows:
given a directed graph $G$, let $G^+$ be the graph obtained by adding two nodes $s$ and $t$ to $G$, 
and for all $v\in G$, adding an edge $(s,v)$ and an edge $(v,t)$ to $G$, plus the edge $(s,t)$.

Finally, for a directed graph $G$, we define its signed adjacency matrix $M(G)$ as follows: 
$M(G)_{ij}$ is $0$ if $i=j$, $+1$ if $(i,j)\in G$, and $-1$ otherwise. 
This allows us to formally state the first, and most important result: 

\begin{theorem}\label{thm:rpd=msr}
	For all $G$, $\rpd(G)=\rpd(G^+) = \msr(M(G^+))-1$.
\end{theorem}

To show this, we first constructively show that any RP realization in $d$ dimensions implies a sign-rank realization with rank $d+1$.
Next, we argue that a sign-rank realization of $M(G^+)$ with rank $r$ implies a weak RP realization in $r-1$ dimensions. 
Using the following lemma, we conclude the theorem:

\begin{lemma}\label{lem:weak=strong}
	A graph $G$ admits an RP realization with vectors in $\R^d$ if and only if it admits a weak RP realization with vectors in $\R^d$.
\end{lemma}

As a special case, we consider the case of partially ordered sets. 

We begin by showing first that the order-dimension of the poset is an upper-bound on the RP dimension.

\begin{theorem}\label{thm:rpd-leq-dim}
	Let $(S,\succ)$ be a poset with order-dimension $k$, a
	and let $G=(S,\succ)$ be the associated directed graph.
	Then $\rpd(G)\leq k$. 
\end{theorem}

This is shown by first introducing the ``dominance ordering'' interpretation of order dimension, 
and then embedding the elements of the partial order in such a way as to achieve such an ordering.
A dominance ordering is simply a partial order on points in $\R^k$, where $\bm x\succ\bm y$ if and only if
$x_i\geq y_i$ for all $1\leq i\leq k$.
This naturally has order dimension at most $k$, since it suffices to sort the points by their positing in each of the $k$ dimensions.

Unfortunately, the converse does not hold. 
We show first that for $k$ sufficiently small, order dimension and RP dimension are equal,
but bad examples exist for large order dimension.

\begin{theorem} \label{thm:rpd-bad-examples}
	Let $(S,\succ)$ be a poset with order-dimension $k$, a
	and let $G=(S,\succ)$ be the associated directed graph.
	Then $\rpd(G)\geq \min\{k,3\}$.
	Furthermore, for all $k\geq 3$, there exist posets $G=(S,\succ)$ of order-dimension $k$, but $\rpd(G)=3$. 
\end{theorem}


\section{RP Dimension and Sign Rank -- Proof of Theorem~\ref{thm:rpd=msr}}\label{sec:thm-msr-proof}\label{sec:msr-proof}
In this section, we prove Theorem~\ref{thm:rpd=msr}, as was outlined in Section~\ref{sec:results}.
We begin by showing the equivalence of RP realization and weak RP realization, 
given by definitions~\ref{def:RP-realization} and~\ref{def:weak-RP-realization}, respectively.
Thus, we wish to prove Lemma~\ref{lem:weak=strong}.
The proof of this lemma effectively reduces to transforming a cone into the cone spanned by the standard basis vectors,
though we have included a careful analysis, which allows us to greatly reduce the technical burden of proving Theorem~\ref{thm:rpd=msr}.

\setcounter{lemma}{0}
\begin{lemma}
	A graph $G$ admits an RP realization with vectors in $\R^d$ if and only if it admits a weak RP realization with vectors in $\R^d$.
\end{lemma}
\begin{proof}
	An RP realization is by definition a weak RP realization, 
	so one direction of the implication follows trivially.
	It suffices to show that the existence of a weak RP realization implies the existence of an RP realization,
	which we will do constructively.
	Let $\vecset:=\{(\bm p_1,\bm x_1),\,\dotsc,\,(\bm p_n,\bm x_n)\}$ be the
	$d$-dimensional vectors which {\em weakly} realize~$G$.
	In other words, $(i,j)\in G$ if and only if $\ip{\bm p_i,\bm x_i}>\ip{\bm p_i,\bm x_j}$.
	Let $\bm b_1, \bm b_2,\dotsc, \bm b_{d-1}$ be a basis for the space orthogonal to the all-ones vector $\bm 1$.
	Thus, $\ip{\bm 1,\bm b_i}=0$ for all $i$, and $\mathcal B = \{\bm 1,\bm b_1,\bm b_2,\dotsc,\bm b_{d-1}\}$ is a basis for $\R^d$.

	We wish to express the $\bm p$ vectors as positive combinations of the $\bm b$ vectors, and thus restrict them to lie in a cone. 
	This will allow us to map the rays of the cone to the standard basis vectors, and get the desired ``strong'' RP realization.
	Since $\mathcal B$ is a basis, we can express $\bm p_i = \sum_{j=0}^{d-1} \alpha^i_j \bm b_j$ for all $i$ and $j$, 
	where $\bm b_0=\bm 1$.
	We have chosen the $\bm b_i$'s as orthogonal to $\bm 1$, and by assumption $\ip{\bm 1,\bm p_i}>0$ for all~$i$.
	Hence, $\alpha^i_0>0$ for all~$i$, and we define $\epsilon = \min\{\alpha^1_0,\dotsc,\alpha^n_0\}>0$.
	However, for all~$j\neq 0$, we may have $\alpha^i_j$  negative.
	To this end, define $\lambda_j = \min\{-1,\alpha^1_j,\alpha^2_j,\dotsc,\alpha^{n}_j\}$ for all $1\leq j\leq d-1$.
	Hence, $\alpha_j^i - \lambda_j \geq 0$ for all~$i$ and~$j$.
	We now define a slightly modified basis $\widehat {\mathcal B }= \{\bm b_1,\dotsc,\bm b_{d-1},\widehat{\bm b_d}\}$,
	where $\widehat{\bm b_d} = \epsilon\bm 1 + \sum_{j=1}^{d-1} \lambda_j \bm b_j$.
	In this new basis, we can express
	\begin{equation}
		\bm p_i = \tfrac{\alpha^i_0}{\epsilon} \cdot \widehat{\bm b_d} + \sum_{j=1}^{d-1} \bm b_j\cdot\left(
			\alpha^i_j - \tfrac{\alpha_0^i\lambda_j}{\epsilon}
		\right)
	\end{equation}
	Recall that $\lambda_j<0$ and $\alpha^i_0\geq \epsilon$, so $\tfrac{\alpha^i_0\lambda_j}{\epsilon}\leq \lambda_j$,
	and thus $\alpha^i_j - \tfrac{\alpha_0^i\lambda_j}{\epsilon}\geq \alpha^i_j-\lambda_j>0$, by construction.
	Thus, not only is $\widehat{\mathcal B}$ a basis for $\R^d$, but the $\bm p_i$ vectors are {\em non-negative} 
	combinations of the basis vectors.
	
	It remains, then to construct a linear map which goes between the standard basis and the basis $\widehat{\mathcal B}$.
	Let $B$ be the matrix whose columns are the vectors of $\widehat{\mathcal B}$, 
	and note that for the $j$-th standard basis vector $\bm e_j$, we have $B\bm e_j = \bm b_j$, for all $1\leq j\leq d$,
	setting $\bm b_d := \widehat{\bm b_d}$.
	Therefore, $B^{-1}\bm b_j = \bm e_j$ for all $1\leq j\leq d$.
	Since we have shown that the $\bm p$ vectors are non-negative combinations of the $\widehat{\mathcal B}$ vectors,
	we may conclude that for all $i$, $B^{-1}\bm p_i$ has all non-negative entries.
	Furthermore, 
	\begin{equation}
	\ip{B^{-1}\bm p_i,B^{\mathsf T}\bm x_j} 
	= (B^{-1}\bm p_i)^{\mathsf{T}}B^{\mathsf{T}}\bm x_j
	= \bm p_i^{\mathsf{T}} (B^{-1})^{\mathsf{T}}B^{\mathsf{T}} \bm x_j 	
	= \ip{\bm p_i,\bm x_j}
	\end{equation}
	Therefore, setting $\widehat{\bm p_i} = B^{-1}\bm p_i$, and $\widehat{\bm x_j} = B^{\mathsf{T}}\bm x_j$, 
	we have that $\ip{\widehat{\bm p_i},\widehat{\bm x_j}} > 0 $ if and only if $\ip{\bm p_i,\bm x_j}>0$, so 
	$\widehat{\vecset}:=\{(\widehat{\bm p_1},\widehat{\bm x_1}),\,\dotsc,\,(\widehat{\bm p_n},\widehat{\bm x_n})\}$ 
	are $d$-dimensional vectors which {\em strongly} realize~$G$, as desired.\qed
\end{proof}

To complete the proof of the theorem, 
it remains to construct low-rank sign matrices for preference graphs which have low-dimensional RP realizations,
and construct low-dimensional {\em weak} RP realizations when the {\em augmented} directed graph has a sign-incidence matrix 
with low sign rank.
We begin by recalling the definition of the augmented preference graph: 
For any directed graph $G$, let $G^+$ be constructed by appending two nodes $s$ and $t$ to $G$, 
adding the directed edge $(s,t)$,
and for all $v\in G$, adding the directed edges $(s,v)$ and $(v,t)$. 
We begin by observing that the addition of these two extra nodes does not affect the RP~dimension of the graph:
\begin{claim} $\rpd(G) = \rpd(G^+)$.
\end{claim}
\begin{proof}
	Let $d := \rpd(G)$ and $d' := \rpd(G^+)$.
	Clearly, $\rpd(G)\leq d'$, since it suffices to remove the vectors representing the $s$ and $t$ nodes from any realization of $G^+$ in $d'$ dimensions.
	It remains to show $\rpd(G^+)\leq d$.  
	We say a vector $\bm x=(x_1,\dotsc,x_d)$ {\em dominates} $\bm y=(y_1,\dotsc,y_d)$ if $x_i\geq y_i$ for all $1\leq i\leq d$.
	Now, let $\vecset:=\{(\bm p_1,\bm x_1),\,\dotsc,\,(\bm p_n,\bm x_n)\}$ be the
	$d$-dimensional vectors which realize $G$.
	Assume that the realization is a standard realization, as in definition~\ref{def:RP-realization}, as opposed to a weak one.
	Then we must have that if $\bm x_i$ dominates $\bm x_j$, there is an $(i,j)$ edge in $G$.
	
	Now, the collection $\bm x_1,\dotsc,\bm x_n$ is finite, and so there must exist vectors $\bm x_s$ and $\bm x_t$ such that
	$\bm x_s$ dominates $\bm x_i$ and $\bm x_i$ dominates $\bm x_t$ for all $1\leq i\leq n$.
	It suffices to set $\bm p_s=\bm p_t=\bm 1$, and this gives a $d$-dimensional realization of $G^+$, as desired.\qed
\end{proof}

Now that we have shown that $G$ and $G^+$ have the same sign-rank, we may introduce the signed adjacency matrix.
For any graph $G$ on the vertex set $\nset s$, let $M(G)$ be defined as 
\begin{equation}
	M(G)_{ij} = \begin{cases}
		0&\text{ if }i = j\\
			+ &\text{ if }(i,j)\in G\\
			- &\text{ if }(i,j)\notin G
	\end{cases}
\end{equation}
In what follows, we show that if $\rpd(G^+)=d$, then \mbox{$\msr(M(G^+))\leq d+1$},
and if $\msr(M(G^+))=r$, then there is a weak RP realization for $G^+$ in $r-1$ dimensions.
Both of those directions will be shown constructively, following a similar construction.
We begin by showing the first direction:
\begin{lemma}\label{lem:rpd-at-least-msr}
	$\msr(M(G^+))\leq \rpd(G^+)+1$.
\end{lemma}
\begin{proof}
	Let $\vecset:=\{(\bm p_1,\bm x_1),\,\dotsc,\,(\bm p_n,\bm x_n)\}$ be the
	$d$-dimensional vectors which realize~$G^+$, that is for all $i,j\in V(G)$, 
	$(i,j)\in E(G)$ if and only if $\ip{\bm p_i,\bm x_i}>\ip{\bm p_i,\bm x_j}$.
	We construct the following matrix:
	let $A(\vecset)$ be the $n\times n$-dimensional matrix whose entries are given by 
	$A(\vecset)_{ij} = \ip{\bm p_i,\bm x_i-\bm x_j}$ for all $i,j\leq n$.
	Observe that we have chosen the entries of $M(G)$ to be exactly the signs of the entries of $A(\vecset)$.
	Thus, it suffices to show that $A(\vecset)$ has rank at most $d+1$, which will imply that $M(G)$ has sign-rank at most $d+1$.
	Indeed, 
	\begin{equation}
		A(\vecset) = \underbrace{\left[\begin{array}{ccc|c}
			\leftarrow&\bm p_1&\rightarrow&\ip{\bm p_1,\bm x_1}\\
			\leftarrow&\bm p_2&\rightarrow&\ip{\bm p_2,\bm x_2}\\
			&\vdots&&\vdots\\
			\leftarrow&\bm p_n&\rightarrow&\ip{\bm p_n,\bm x_n}
		\end{array}\right]}_{n\times(d+1)}\cdot
		\underbrace{\left[\begin{array}{cccc}
			\uparrow&\uparrow&&\uparrow\\
			-\bm x_1&-\bm x_2&\dotsm&-\bm x_n\\
			\downarrow&\downarrow&&\downarrow\\[3pt]\hline
			1&1&\dotsm&1
		\end{array}\right]}_{(d+1)\times n} \label{eq:A(X)}
	\end{equation}
	Since the inner-dimension of the product is $d+1$, this implies that $A(\vecset)$ has rank at most $d+1$, 
	as desired. \qed
\end{proof}

We will use this same construction to show the converse.
The extension $G^+$ is required to ensure that this is possible, and that the vectors do indeed form a weak RP realization.

\begin{lemma}\label{lem:rpd-at-most-msr}
$\rpd(G^+) \leq \msr(M(G^+)) -1$.
\end{lemma}
\begin{proof}
Let $A$ be some rank-$r$ realization of $M(G^+)$.
Assume without loss of generality that the first and second rows and columns of $M(G^+)$ are associated to the dominating and dominated vertices, respectively. Thus, $M(G^+)$ has the form
\[
	\left[\begin{matrix}
		0 & + & + & + & + & \dotsm\\
		- & 0 & - & - & - & \dotsm\\
		- & + & 0 & * & * & \dotsm\\
		- & + & * & 0 & * & \dotsm\\
		\vdots & \vdots& \vdots& \vdots & \ddots & \ddots
	\end{matrix}\right]
\]
Thus, letting $A_i$ be the $i$-th row of $A$, we have that $A_1-A_2$ is an all-positive vector.
Since $A$ has rank $r$, we may set $\bm a_1=A_1-A_2$, and extend it to a basis $\mathcal A = \{\bm a_1,\bm a_2,\dotsc,\bm a_r\}$
for the rows of $A$.
Let $R$ be the $r\times n$ matrix whose rows are the vectors of $\mathcal A$, and let $L$ be the matrix of coefficients such that 
$A=LR$.
Note that scaling the columns of $A$ is the same as scaling the columns of $R$, 
and this scaling process does not affect the rank of the matrix.
Furthermore, if the scaling factors are positive, then the sign pattern is unaffected.
Thus, we may assume without loss of generality that $\bm a_1 = \bm 1$, by rescaling column $j$ by $1/(\bm a_1[j])>0$,
for all $j$.
(We are using square brackets to denote the entries of the vector.)
Thus, we may interpret the entries of $L$ and $R$ as in equation~\eqref{eq:A(X)}. 
Now, if the $i$-th row of $L$ is given by the vector $\bm \ell_i$,  we set $\bm p_i := \bm \ell_i[1..r-1]$,
and assume $\ip{\bm p_i,\bm x_i} = \bm \ell_i[r]$.
Furthermore, if the $j$-th column of $R$ is given by $(\bm r_j,1)$, then we set $\bm x_j = -\bm r_j$.
Since the diagonal entries of $M(G^+)$ are zero, we must have that $1\cdot \ip{\bm p_i,\bm x_i} + \ip{\bm p_i,-\bm x_i}=0$,
which is consistent.

Thus, we have vectors $\bm p_i$ and $\bm x_j$ in $\R^{r-1}$ such that $\ip{\bm p_i,\bm x_i}>\ip{\bm p_i,\bm x_j}$ if and only if $(i,j)\in G^+$.
It suffices to transform these vectors to ensure $\ip{\bm 1,\bm p_i}>0$ for all $i$.
Recall that we have assumed that $G^+$ contains the edges $(1,i)$ and $(i,2)$ for all $3\leq i\leq n$.
Thus, we must have $\ip{\bm p_i,\bm x_i}\leq \ip{\bm p_i,\bm x_1}$ 
and $\ip{\bm p_i,\bm x_i}> \ip{\bm p_i,\bm x_2}$.
Therefore, $\ip{\bm p_i,\bm x_1-\bm x_2}>0$ for all $i\geq 3$.
Similarly to the proof of Lemma~\ref{lem:weak=strong}, we will use this to find an appropriate linear transformation for the $\bm p$ and $\bm x$ vectors.
Let $Q$ be any invertible matrix such that $Q(\bm x_1-\bm x_2) = \bm 1$.
Then, setting $\widehat{\bm p_i} := (Q^{-1})^{\mathsf T} \bm p_i$, and $\widehat{\bm x_j} := Q\bm x_j$,
we have $\ip{\bm p_i,\bm x_j} = \ip{\widehat{\bm p_i},\widehat{\bm x_j}}$, and $0<\ip{\widehat{\bm p_i},\widehat{\bm x_1}-\widehat{\bm x_2}}=\ip{\widehat{\bm p_i},\bm 1}$.
Thus, we have constructed a weak RP realization for $G^+$ in $r-1$ dimensions.\qed
\end{proof}

We claim that this completes the proof of Theorem~\ref{thm:rpd=msr}: 
Lemma~\ref{lem:weak=strong} ensures that a weak realization is possible if and only if a ``strong'' one is,
the above Claim ensures that $\rpd(G^+)=\rpd(G)$,
these two facts along with Lemma~\ref{lem:rpd-at-least-msr} imply that $\rpd(G)\leq \msr(M(G^+))-1$,
and finally, Lemma~\ref{lem:rpd-at-most-msr} implies that $\rpd(G)\geq \msr(M(G^+))-1$,
from which we conclude Theorem~\ref{thm:rpd=msr}.

\section{RP Dimension and Order Dimension -- Proofs of Theorems~\ref{thm:rpd-leq-dim} and~\ref{thm:rpd-bad-examples}}
\label{sec:poset-proofs}
In this section, we prove Theorems~\ref{thm:rpd-leq-dim} and~\ref{thm:rpd-bad-examples}, as was outlined in Section~\ref{sec:results}.
We begin by defining the notion of a dominance order, and noting the natural interpretation of dimension as order dimension.
Recall from section~\ref{sec:thm-msr-proof} the notion of vector dominance: where we say $\bm x$ dominates $\bm y$ 
if it is at least as great in each coordinate. 
This is denoted $\bm x\geq \bm y$.
A standard form of geometrically-defined partial orders is a vector-dominance partial order: 
Given a set of vectors $\bm x_1,\,\dotsc,\,\bm x_n$, we set $i\succ j$ if and only if $\bm x_i\geq \bm x_j$.
It is easy to check that this relation is transitive and acyclic.
We remark that the vector-dominance poset induced by points in $\R^d$ has order dimension at most $d$:
For all $1\leq j\leq d$, set the $j$-th total order to be the ordering of the $n$ points with respect to their $j$-th coordinate.
Then $i\succ j$ if and only if all $d$ total orders agree on the relative ordering of $\bm x_i$ and $\bm x_j$.
The converse also holds: if a partial order has order-dimension $k$,
then it can be expressed as a vector-dominance poset in $\R^k$. 
We will formalize this fact and extend it to show Theorem~\ref{thm:rpd-leq-dim}.

\begin{proof}[of Theorem~\ref{thm:rpd-leq-dim}]
	Let $(S,\succ)$ be a poset with order-dimension $k$, a
	and let $G=(S,\succ)$ be the associated directed graph.
	We wish to show $\rpd(G)\leq k$.

Let $\succ_1,\,\dotsc,\,\succ_k$ be the $k$ total orders which realize $\succ$.
	Note that each total order $\succ_i$ induces a ranking $\sigma_i$ on the elements of 
	$S$, such that $a\succ_i b$ if and only if $\sigma_i(a)>\sigma_i(b)$.
	Assume without loss of generality that $\sigma_i$ maps the elements of $S$ to $\nset{|S|}$.
	Then the usual dominance embedding of $(S,\succ)$ is given by mapping each $a\in S$
	to $\phi(a) := \bm \sigma(a) := (\sigma_1(a),\sigma_2(a),\dotsc,\sigma_k(a))$.
	Now, $a\succ b$ if and only if $\phi(a)$ dominates $\phi(b)$.
	
	For the purposes of RP-dimension, we need a rescaled embedding.
	If $k=1$, then $(S,\succ)$ is a total order, and setting $\bm x_i=\sigma(i)$,
	$\bm p_i=1$ will suffice.
	Otherwise, define $\psi(a):= (k^{\sigma_1(a)},k^{\sigma_2(a)},\dotsc,k^{\sigma_k(a)})$.
	Since $k\geq 2$, this maintains the dominance ordering of $\phi$.
	Now, if we set $\bm x_i=\psi(i)$, and 
	\[
		\bm p_i = \left(
			\tfrac1{k^{\sigma_1(i)}},\tfrac1{k^{\sigma_2(i)}},\dotsc,\tfrac1{k^{\sigma_k(i)}}
		\right)
	\]
	then we get $\ip{\bm p_i,\bm x_i}=k$. 
	Furthermore, letting $\bm y^i_j = (\dotsc,0,k^{\sigma_j(i)+1},0,\dotsc)$, 
	we have $\ip{\bm p_i,\bm y^i_j} = k$ for all $1\leq j\leq k$.
	Therefore, the hyperplane normal to $\bm p_i$, passing through $\bm x_i$,
	will also pass through $\bm y^i_j$ for all $j$.
	Since the $\sigma$ values are presumed to be positive integers, this means that 
	$\bm x_i=\psi(i)$ can only be revealed-preferred to $\bm x_j=\psi(j)$ if $\bm x_i$
	dominates $\bm x_j$.
	
	Thus, we have constructed a set of vectors in $\R^k$ which realizes $G=(S,\succ)$,
	which allows us to conclude that $\rpd(G)\leq k$. \qed
\end{proof}

It remains to prove Theorem~\ref{thm:rpd-bad-examples}, that is that for $k=1$, $2$, or $3$,
posets of order dimension $k$ have RP dimension $k$, 
but that for all $k\geq 3$, there exists a poset of order dimension $k$ but order dimension $3$.
We begin by showing this first part: 
\begin{lemma}
	A poset has RP dimension 1 or 2 if and only if it has order dimension 1 or 2, respectively.
\end{lemma}
\begin{proof}
	Note that an RP realization in $\R^1$ is simply a total order on the players, since we have 
	$(i,j)\in G$ if and only if $p_ix_i> p_ix_j$. 
	Since we need $p_i > 0$, we have $(i,j)\in G$ if and only if $x_i>x_j$, and therefore,
	the values of $x_1,\,\dotsc,\,x_n$ induce a total order on the elements.
	Thus, any graph has RP dimension 1 if and only if it is a poset of order dimension 1, 
	namely, a total order.
	
	Thus, we conclude that a poset of order dimension 2 must have RP dimension equal to 2.
	It remains to show that a poset with RP dimension 2 must have order dimension 2.
	It is known that a poset has order dimension 2 if and only if the complement of its {\em comparability graph} is also a comparability graph~\cite{posetOrder2}.
	In our terms, a poset $(S,\succ)$ has order dimension 2 if and only if (a) at least one pair of elements is not comparable,
	{\em i.e.} there is some $x,y\in S$ such that neither $x\succ y$ nor $y\succ x$,
	and (b) there exists a partial order $\succ'$ on $S$ whose comparable pairs are exactly those which are non-comparable in $(S,\succ)$.
	Thus, for any two $x,y\in S$, we must have exactly one of $x\prec y$, $x\succ y$, $x\prec' y$, and $x\succ' y$ hold.
	Therefore, to show that a partial order with RP dimension 2 must have order dimension 2, 
	we must construct a partial order on its non-comparable pairs.
	
	Let $\vecset:=\{(\bm p_1,\bm x_1),\,\dotsc,\,(\bm p_n,\bm x_n)\}$ be the
	$2$-dimensional vectors which realize~$(S,\succ)$, that is for all $i,j\in S$, 
	$i\succ j$ if and only if $\ip{\bm p_i,\bm x_i}>\ip{\bm p_i,\bm x_j}$.
	Furthermore, denote $\bm x_i = (x_i,y_i)$ and $\bm p_i = (p_i,q_i)$ for all $1\leq i\leq n$.
	Recall also that we have assumed $\bm p^i_j\geq 0$ for all $1\leq i\leq n$ and $j=1,2$.
	For every pair $i,j$ such that neither $i\succ j$ nor $j\succ i$, 
	we say $i\succ' j$ if $x_i>x_j$, and $j\succ' i$ otherwise.
	This is clearly acyclic, it remains to show that the relation is transitive.
	
	Let $\bm x_1,\bm p_1,\bm x_2,\bm p_2,\bm x_3,\bm p_3$ be such that 
	$\ip{\bm p_i,\bm x_j}\geq \ip{\bm p_i,\bm x_i}$ for all $j\neq i$,
	and $x_1>x_2>x_3$. 
	This implies that $1\succ' 2\succ' 3$. We wish to show that $1\succ' 3$.
	It is clear that $x_1>x_3$, so it remains to show that both $\ip{\bm p_1,\bm x_3}\geq \ip{\bm p_1,\bm x_1}$,
	and $\ip{\bm p_3,\bm x_1}\geq \ip{\bm p_3,\bm x_3}$.	
	We may assume without loss of generality that $p_i>0$ for $i=1,2,3$, since we may slightly rotate the space,
	and so we may assume without loss of generality that $p_i = 1$ for $i=1,2,3$, since scaling $\bm p$ does not affect the induced preferences.
	Thus, $\ip{\bm p_i,\bm x_j} = x_j + q_iy_j$ for all $i, j$.
	Now, $x_1>x_2$, but $x_2+q_1y_2>x_1+q_1y_1$, so we must have $y_1<y_2$.
	Similarly, we get $y_2<y_3$.
	Furthermore, we have
	\begin{equation}
		\begin{array}{rrl}
			&x_1+q_2y_1 &> x_2 + q_2 y_2\\
			-&x_1+q_1y_1 &< x_2 + q_1 y_2\\\hline
			\Rightarrow& (q_2-q_1)y_1 &> (q_2-q_1)y_2
		\end{array}
	\end{equation}
	but $y_1<y_2$, so we conclude $q_2<q_1$. Similarly, we have $q_3<q_2$.
	Now, since $q_1>q_2$ but $y_2<y_3$, we have that 
	\begin{equation}
		x_3 + q_2y_3 \geq x_2 + q_2y_2 \quad \implies\quad
		x_3 + q_1y_3 \geq x_2 + q_1y_2
	\end{equation}
	But we know that $x_2+q_1y_2 \geq x_1+q_1y_1$, so we have that $\ip{\bm p_1,\bm x_3}\geq \ip{\bm p_1,\bm x_1}$.
	The converse inequality is shown similarly, and thus we may conclude that the relation $\succ'$ is a partial order.

	Since any two elements are comparable in $\succ$ if and only if they are not comparable in $\succ'$,
	we conclude that the order dimension of $\succ$ is at most 2.
	Since its RP dimension is not 1, it must have order dimension exactly 2, as desired.\qed
\end{proof}

With this lemma in hand, we conclude that if a poset has order dimension 3, it must have RP dimension 3,
as its RP dimension is at most 3, but it cannot have dimension 1 or 2. 
Thus, we have proved the first half of Theorem~\ref{thm:rpd-bad-examples}.
It remains to show that for all $k\geq 3$, there exist posets with order dimension $k$, but RP dimension 3.

The family of {\em standard} posets $S_2,\,S_3,\,\dotsc$ is a sequence of posets such that
$S_k$ has order dimension $k$, and ground set of size $2k$. (See {\em e.g.} \cite{standardPoset}) 
They are defined as follows: 
The ground set for $S_k$ is labelled $1,\,2,\,\dotsc,\,k,\,1',\,2',\,\dotsc,\,k'$,
and we have $i'\succ j$ for all $i\neq j$. No pair $i',j'$ or $i,j$ is comparable.

We will show that $G_k = (S_k,\succ_k)$ has RP dimension 3 for all $k\geq 3$.
Let $H:=\{\bm x\in \R^3: \ip{\bm 1,\bm x}=0\}$, the plane normal to the all-ones vector in $\R^3$,
and let $\mathbb S_H$ be the unit circle in $H$ centered at the origin.
Thus, $\mathbb S_H = \{\bm x\in H:\Vert \bm x\Vert_2 = 1\}$.
Finally, let $\bm a_1,\,\dotsc,\,\bm a_n$ be $n$ equally spaced points along the circumference of~$\mathbb S_H$.
We will use these to construct our realization of $(S_k,\succ_k)$ in $\R^3$.
For all $i\leq k$, set $\bm x_i = (2+\epsilon)\bm a_i$, and $\bm x_{i'} = \bm 1 - \bm a_i$, where $\epsilon>0$ will be chosen later.
Set $\bm p_i = \bm 1 - \bm a_i$, and $\bm p_{i'} = \bm 1 + \bm a_i$.
Since the $\bm a_i$'s are unit vectors, we have that $\ip{\bm a_i,\bm a_j}<1$ if $i\neq j$, and $=1$ if $i=j$.

We have $\ip{\bm p_i,\bm x_j} = (2+\epsilon)\ip{\bm 1,\bm a_j} - (2+\epsilon)\ip{\bm a_i,\bm a_j}$.
The left hand term is 0, and the right hand term is minimized when $i=j$.
Thus, $\bm x_i$ is not revealed preferred to $\bm x_j$ for all $j\neq i$.
Furthermore, $\ip{\bm p_i,\bm x_{j'}} = \ip{\bm 1-\bm a_i,\bm 1-\bm a_j} = \ip{\bm 1,\bm 1} + \ip{\bm a_i,\bm a_j}$,
since $\ip{\bm 1,\bm a_j}=0$ for all $j$.
Thus, $\bm x_i$ is not revealed preferred to $\bm x_{j'}$ for all $j$.

Now, $\ip{\bm p_{i'},\bm x_{i'}} = \ip{\bm 1+ \bm a_i,\bm 1-\bm a_i} = \ip{\bm 1,\bm 1} - \ip{\bm a_i,\bm a_i} = 2$,
whereas $\ip{\bm p_{i'},\bm x_j} = (2+\epsilon)\ip{\bm 1+\bm a_i,\bm a_j} = (2+\epsilon)\ip{\bm a_i,\bm a_j}$.
Thus, if we choose $\epsilon>0$ sufficiently small, we have that $\bm x_{i'}$ is revealed preferred to $\bm x_j$ for all $j\neq i$, but not to $\bm x_i$.
Furthermore, $\ip{\bm p_{i'},\bm x_{j'}} = 3 - \ip{\bm a_i,\bm a_j}$, which is minimized when $i=j$, so $\bm x_{i'}$ is not revealed preferred to $\bm x_{j'}$ for all $j\neq i$.

Therefore, we have shown that our choice of $\bm p_i$'s and $\bm x_j$'s is a valid RP realization of $(S_k,\succ_k)$ in $\R^3$ for all $k\geq 3$.
Thus, we have demonstrated the existence of partial orders with order dimension $k$ but RP dimension 3, for all $k\geq 3$,
hence concluding the proof of Theorem~\ref{thm:rpd-bad-examples}.

\section{Further Work}

This paper does not address the computational complexity of computing the RP dimension of a given graph,
and this is left as an open problem for future work.
Below is a summary of the computational complexity of matrix sign rank, and what this implies for RP dimension.

\bigskip
\noindent \textbf{Complexity of Matrix Sign Rank.} Recall that the problem of computing RP dimension is denoted \revprefdim, and the problem of computing matrix sign rank, \matsgnrank.
It is known~\cite{BK15} that \matsgnrank in full generality is complete for the {\em existential theory of the reals}:
the problem of determining whether a system of polynomial equalities and inequalities has a feasible solution over the reals.
This complexity class, often denoted ``$\exists\R$'', is known to lie between \textsf{NP} and \textsf{PSPACE}.
In fact, it is $\exists\R$-complete to determine whether a matrix has sign rank at most 3. 
However, this hardness result only holds when the sign matrix is allowed to have arbitrarily many zero entries in each row and column. 
When sign matrices are constrained to have no 0 entries, \matsgnrank is known only to be \textsf{NP}-hard. 
(Again,~\cite{BK15}).
It is not known whether \matsgnrank lies in \textsf{NP} in this restricted setting,
though we think this is unlikely.
\bigskip

This paper shows that the \revprefdim problem is equivalent to computing the sign rank of signed adjacency matrices,
which are a (large) subclass of sign matrices with exactly one zero in each row and column.
Note that replacing a row containing a single zero entry with two copies, replacing the zero with a $+$ in one copy, and a $-$ in the other, does not affect the sign rank. 
Therefore, \revprefdim is a special case of \matsgnrank in the restricted setting, 
and thus it cannot be a harder problem.

We leave as an open problem determining whether \revprefdim is itself \textsf{NP}-hard, and whether it is equivalent to \matsgnrank on $+,-$ matrices.

%

 \bibliographystyle{splncs04}
 \bibliography{RP-citations}

\color{gray}
\end{document}